%% file: groupring2.tex
\documentclass[a4paper]{article}
\usepackage{fullpage,tedmath}%authblk}
\newcommand{\tr}{\text{tr}\,}

\title{Group ring cryptography: Cryptographic schemes, key exchange, public key}
\author{Ted
 Hurley\footnote{National Universiy of Ireland Galway, email:
 Ted.Hurley@NuiGalway.ie }}
\date{}

\begin{document}

\maketitle

\begin{abstract} General cryptographic schemes  are presented 
where keys can be one-time or ephemeral. % based on . 
Processes for key exchange are derived. Public key cryptographic schemes
based on the new systems are easily established. Authentication and signature
schemes are  implemented. The schemes are an advance
on group ring techniques and are easily implemented but highly secure.   
They may be integrated with error-correcting  coding schemes   
so that encryption/coding and decryption/decoding may
be done simultaneously.  
%Keys used  may be session keys or used  for a series of transfers. 
  %and after a session 
%  and new schemes are easily derive %for further transfers.
% is then easy to implement a new system. 
%Transfers of data without key exchange involves
%  three transmissions as does a key exchange. 
%Keys are easily generated.  
%There are no `public keys' to be maintained and  generated. 
%Key sizes are small and systems are secure and 
%Running time for encryption and decryption implementation complexity 
%can be $O(n\log n)$.
% Keys are  
%selected randomly from a large available pool.
% from which they may be selected randomly.
%and no keys are `public' at any
%time.  
%Parts of random matrices are not made  public. 
\footnote{AMS Classification Classification: 15A99, 94A60 }\end{abstract}
\section{Introduction}
%Here is an idea for encryption. 
%
This paper introduces cryptographic systems based on operations
with randomly chosen vectors, matrices and group ring elements. These are  an advance on group ring cryptographic techniques \cite{hur5} and are easily implemented and secure. 
Keys used may be one-time 
session  keys or ephemeral;   as they are easily constructed they may be 
changed if required  for each transaction or series of transactions.
Key exchange methods are
derived.  Public key cryptographic schemes based on the new systems
are derived. It is straightforward to include authentication, signature
and `person-in-middle' interference prevention 
methods based on
the schemes. Encryption may be 
incorporated with error-correcting codes, \cite{hur1,hur4,hur8}, so that encryption/coding and
decryption/decoding may be done simultaneously.
A public-key scheme constructed may be altered and made private to an individual. 
% and then messages may then be authenticated and verified. 
% communicate secretly to B and B will
%know the message can only have come from A.  

%Running times  for encryption and
%decryption implementation for many schemes are $O(n\log n)$. 
 
%The methods may be used for secure key exchange. 
%Public
%keys can be initiated which work between two or more correspondents.
%Matrices from a huge available pool are selected randomly. %large matrices. 

%Group ring matrices prove particularly useful and  easy to
%implement. % and secure. 

%Keys are randomly selected. % and constructed.

%Matrices with large kernels only are revealed in the processes. 

Details have also  appeared in \cite{hur6}. % randomly.

Generally the running time for encryption is at most $O(n\log n)$ and for decryption (when the the key is known) the running time  is also at most of this order.  Large pools from which to randomly draw the keys are available; 
 for example, see sections \ref{comm}, \ref{singular}, systems of size $101$  over $Z_p$
%and vectors in $\Z_p^{101}$ 
%which are considered as elements in a
%group ring  
has of  order $p^{100}$ different choices for the construction of a key.

%but when keys are
%exchanged only one is necessary. 

%changed for each transaction or
%after a series of transactions or at time intervals.
 %with no difficulty. 

\subsection*{Features}% of these schemes}

\begin{itemize}
\item  Encryption and decryption keys are easy to construct 
and   can be chosen for a one-time  session 
  or series of transactions. 
\item Key exchange schemes are derived. 
\item Public key cryptographic schemes are developed. These can be
  altered for private communication and messages then authenticated.
\item Authentication, signature
and `person-in-middle' interference prevention 
methods are provided. %Authentication/signature methods can be included.
%\item Secure. 
%\item List
%\item Encryption and decryption can be implemented in $O(n\log n)$ time. 
 %\item Different types of matrices can be
 % combined and used within the same system.
\item Encryption and error-correction coding may be integrated into one system.
Coding and encryption can complement one another. 
%\item What is transmitted (and
 % openly revealed) has large kernel when made up as matrices. 

%\vdots
\end{itemize}

When a system is used for one-time session then three transmissions are
necessary. A key exchange also requires three transmissions but
once a key has been exchanged each transaction naturally then requires
just the one transmission. 

%is used in message exchange. 
%Public keys are made up by each participant. % and are easy to change. 
\subsection*{Layout} The layout of the paper is as follows:
\begin{enumerate}

\item Details on various theory required for the constructions
are given in section \ref{theory} and may be consulted as
required. Here also various systems and schemes 
within which  %in addition %as part of
                                %the general theory   
 the constructions may be realised  %can take place
 are outlined. 
 
\item General encryption methods are
  introduced in section \ref{message}.
\item Key exchange methods are laid out in section \ref{key}.
\item Public key encryption methods are given in section \ref{public}.
\item Methods to include error-correcting with the cryptography is
  presented in section \ref{coding}.
\item Multiple design methods are presented in section
  \ref{mult}.
\item Section \ref{who} derives authentication, signature and
  `person-in-middle' solutions.
\end{enumerate}

%The literature contains an enormous amount on crytography including
%many books. 
Basic references for cryptography 
include  \cite{koblitz}, \cite{koblitz1}, \cite{hand}. 
The first two in particular contain much of the algebra required 
  and further basic algebraic 
material may be obtained in \cite{sehgal}.

\section{Encrypt message}\label{message}
%General schemes are laid out in section \ref{tryy}. 
%Basic methods are introduced initially. 
\subsection{General Schemes}\label{encrypt2} % with not necessarily commuting matrices}
%General cryptographic methods
% for not necessarily commuting matrices  are developed as follows.
Here  $R_{n\ti n}$ denotes the ring of $n\ti n$ matrices,
%over a ring $R$ 
and $R^n$ is the ring of vectors of length $n$, over
a ring $R$.  $RG$ denotes the group ring of the group $G$ over
the ring $R$;
 for details on general properties of group rings see
\cite{sehgal}. The $R$ is usually a field and is often then denoted by
$F$.   
For details on {\em group ring matrices} see 
 for example \cite{hur}; the set-up and main
 properties of these are given in section
\ref{results}. They may also be referred to as   
{\em $RG$-matrices} when the group ring in question is 
specified  %are used; 
%These are matrices of a special
%type and structure and also may be constructed randomly of large
%size. 
%More details on these may be obtained in section \ref{} and in general
%in \cite{hur}. 
and are obtained from the embeddings of  group
rings into rings of matrices \cite{hur}; they include such matrices as circulant
matrices, circulant of circulant matrices and similar. 
 %and also certain types of matrices which do not commute.   
%For a ring $R$ and group $G$, group ring matrices are defined for
%example  in  \cite{hur}. 

An $RG$-matrix,  which is of size $|G|\ti |G|$, is determined by its first row and 
is a matrix corresponding to a group
ring element, relative to a listing of the elements of $G$.
%More details are given in section \ref{results}.  
Two $RG$-matrices obtained from the same group ring 
 $RG$ (relative to the same listing) are said to be of
the same {\em type}. The $RG$-matrices commute if and only if $G$ is
commutative. 
 Methods to randomly choose singular and  non-singular
matrices with certain properties from a huge  pool of such matrices are given 
in section \ref{singular}.  

 Let $\un{x}$ be a row vector with entries from $R$. 
Then the {\em completion} of
$\un{x}$ in $RG$ (relevant to a particular listing) is the
$RG$-matrix with first row $\un{x}$. The {\em rank of a vector},
relative to 
its completion in a specified group ring, is defined
as the rank of its completion; this 
gives meaning to {\em kernel of a vector} relative to its group ring
completion.  

The completion of the vector $\un{x}$  is denoted by the corresponding
capital letter $X$ (without underlining). For $a\in RG$ its image in
$R_{n\ti n}$ under the embedding of $RG$ into $R_{n\ti n}$ is denoted by
the corresponding capital letter $A$.  When $\un{x}$ is a vector to be
considered as an element of $RG$ then also use $X$ (without
underlining) to denote its image under this embedding. 

%When $\un{x}$ is
%considered as an element of a group ring $RG$ then also $X$ will
%denote its image under the embedding of the   
%These $RG$-matrices are special structured matrices. 
The following Lemma is immediate.
\begin{lemma}\label{rg} Suppose $P,Q$ are $RG$-matrices with the same first
  row. Then $P=Q$. 
\end{lemma}

Thus if $\un{x}$ is a vector in $R^n$ and $A$ is an $RG$-matrix of size
$n\ti n$ then from Lemma \ref{rg} the
completion of $\un{x}A$ in $RG$ is $XA$ where $X$ is the completion of
$\un{x}$ in $RG$.  
%\subsubsection{Basic method, non-commuting} % non-commuting matrices}

Let $\un{x}$ be the data to be transmitted secretly from A(lice) to B(ob). 
The data $\un{x}$ is arranged so 
that $X$ is singular with large kernel 
where $X$ is the completion of $\un{x}$ in same type of $RG$-matrix as
$A$ (the matrix chosen by A below in 1.); details on how this can be
arranged are given below in
section \ref{singular}.  When $X$ is singular with large kernel 
then also $CX,XC$ are singular with large kernel for any matrix $C$. 
\subsubsection{General set-up}\label{tryy}
\begin{enumerate}
\item A chooses $A$, a non-singular group ring matrix, and transmits $\un{x}A$. 
%(The completion of  $\un{x}A$ is $XA$.) % as $XA$ is an $RG$-matrix whose first row is $\un{x}A$.
\item B chooses $B$ non-singular and transmits $BXA$. 
\item A transmits $BX$. % (which is $BXAA^{-1}$). %and transmits this.
\item B works out $B^{-1}BX = X$. 
\end{enumerate}
$B$ need not in general be a group ring matrix and even if so it need
not be of the same type as $A$. 
If $B$ is of the same type as
  $A$ and $X$ then only the 
  first row of the matrices in 2.\ and 3.\ need be transmitted.
In 4.\ only the first row of $B^{-1}BX$
  need be calculated as the first row of $X$ give $\un{x}$. The
  inverses of $A,B$ should be easily obtainable; 
pools of
  matrices from which such matrices may be drawn are given in section
  \ref{theory}. When using matrices with certain structures such as
  $RG$-matrices the matrix multiplications and vector-matrix
  multiplications may be performed by convolutional methods. 
%(For example schemes with $A^2=I$ or $A^p=\al I$
%for a scalar $\al$ are attainable.) 

The matrices $A,B$ here are  as large as  the vector
of data $\un{x}$ and chosen randomly. The data may also be broken up
and multiple vector design schemes implemented as shown in section
\ref{mult}. %The non-singular matrices are chosen so that the inverse
%is strightforward to calculate as for example $A^2=I$ or $A^p=\al I$
%for a scalar $\al$. 

Simplified schemes with commuting matrices  are
derived in the next section \ref{encrypt1}; these do not necessarily
need $RG$-matrices.

\subsubsection{Commuting schemes}\label{encrypt1}
Suppose the large pool of matrices available commute with one
another. %Methods for choosing random such singular and non-singular
%matrices appear in section \ref{theory}. 
In these cases simplified schemes may be designed as follows. 
% the singular ones can be
%chosen with large kernel and used in transmission.

Let $\un{x}$ be data to be transmitted secretly from A(lice) to
B(ob). 
\begin{enumerate} \item A chooses  the matrix
$A$ non-singular and $\un{x}A$ is transmitted. %($A$ is non-singular.)
\item B chooses the matrix $B$ non-singular 
and transmits $\un{x}AB$. %($B$ is non-singular.)
\item  A transmits $\un{x}ABA^{-1}=\un{x} B$.
 %since $AB=BA$, and  transmits this. %This is transmitted.
\item B applies $B^{-1}$ to $\un{x}B$ to get $\un{x}$.
\end{enumerate}

Even when the matrices commute, the general scheme of \ref{tryy} may
 still be used.

%The matrices here are chosen from a pool of commuting matrices. 
Schemes with commuting matrices may be achieved using group ring 
matrices derived from  an abelian group ring as for example circulant
matrices or circulant of circulant matrices.
 Section \ref{comm} discusses types of
such  matrices which may be used.   
 When such group ring matrices are used the data
is  arranged so that the group ring 
matrix formed using $\un{x}$ as first row
is singular with large kernel; how to arrange the data in such a way
is discussed later. 
%The matrices $A,B$ are as large as  the vector
%of data $\un{x}$ and chosen randomly. The data may also be broken up
%and multiple vector design schemes implemented as shown in section
%\ref{mult}. 
The non-singular matrices are chosen so that the inverses
are immediate or straightforward to calculate.
 %for example schemes with $A^2=I$ or $A^p=\al I$
%or a scalar $\al$ are available. % from a huge pool.  

When some matrices commute the data $\un{x}$ may also be `protected' at each
end as follows:
\begin{enumerate}
\item A chooses $A_1,A_2$ non-singular and transmits
  $A_1XA_2$.  
%(The completion of  $\un{x}A$ is $XA$.) % as $XA$ is an $RG$-matrix whose first row is $\un{x}A$.
\item B chooses $B_1,B_2$ non-singular and transmits
  $B_1A_1XA_2B_2$. It is necessary that $A_iB_i=B_iA_i$ for $i=1,2$
  but otherwise there are no commuting conditions. 
\item A transmits $B_1XB_2$. % (which is $BXAA^{-1}$). %and transmits this.
\item B works out $ X$. 
\end{enumerate}
If all matrices commute, including $X$, then the system is the same as
above, with $A$ replaced by $A_1A_2$ and $B$
replaced by $B_1B_2$. 
%The matrix multiplcation can be done by convolution methods. 

\section{Key exchange}\label{key}
A modification of the general scheme is now  set up  so that  
a process may be initiated whereby two intended correspondents can exchange
a secret encoder/decoder.

Let $\{\un{x}, \un{y}\}$ be vectors so that $\{X,Y\}$ are singular with large
kernel and  some combination of
$\{X,Y\}$ or some combination of $\{X,Y\}$ with a known element or elements 
%(as for example the identity or half-way element)
 is non-singular. Methods to randomly  choose such vectors
 $\{\un{x},\un{y}\}$ are developed in section \ref{theory} below. 
\begin{enumerate}
\item A chooses $A$ non-singular and $\un{x}$ with large kernel 
and transmits $\un{x}A$. 
\item  B chooses $B$ non-singular and transmits $BXA$.
\item A transmits $BX$. B now knows $X$. A can now repeat the process to get $Y$
  secretly to B. Or else:
\begin{enumerate}
\item B chooses $Y$ with large kernel and $B_1$ non-singular 
so that a combination of $\{X,Y\}$ 
  with a known element or known elements is
  non-singular and transmits $B_1Y$.
\item A transmits $B_1YA$.
\item B transmits $YA$.
\end{enumerate}   
\item Both A and B now have $X,Y$ from which to form the encoding
  matrix for use between A and B.
\end{enumerate}

%\subsection*{Further}
 When $\un{x},\un{y}$ are known, an $RG$-matrix may be formed
 from these using a { different} $RG$ from that used for the key
 exchange. Convolutional methods where appropriate as with group ring
 matrices may be used for matrix and matrix-vector
 multiplications.  
 It is sometimes the case that it is sufficient to simply add on a known element
or known elements  to $\un{x}+\un{y}$ to obtain an element whose
completion is non-singular. For example  
$X+Y+1$ may be known to be non-singular in some systems, see section
\ref{theory}. 
Knowing the added element(s) gives no information as
$\un{x},\un{y}$ are known only to A,B.  In  section \ref{idemspot} it
is shown how to randomly choose $\{X,Y\}$ each with large kernel so
that a (linear) 
combination of $\{X,Y\}$ is non-singular and its inverse is easily
constructed. 

In cases $X,Y$ may be chosen so that small powers of $X,Y$ are
zero. This ensures that
$\ker X, \ker Y$ are large, see Corollary \ref{rank1} below,  and then
$\ker XC, \ker CY, \ker CX, \ker YC$ are also  
large for any matrix $C$. 
%$B_1$ should be of same type as the completion of $\un{x}$ and $A_1$
%should be of the same type as the completion of $\un{y}$.

When key has been
exchanged  between A and B, messages between them  may then be encrypted
directly. When a key has been exchanged
it is not necessary to arrange data to
be transferred to have large kernel.  
Messages may also then be encrypted and encoded  together as shown
later. 

 The data $\un{x},\un{y}$ may be protected on both sides when some
 matrices commute; in 1.\ above, A chooses $A_1,A_2$ non-singular and
 $\un{x}$ with large kernel and transmits $A_1XA_2$ to which B chooses
 $B_1,B_2$ where $A_iB_i=B_iA_i$ 
and transmits $B_1A_1XA_2B_2$ and process continues as above.
 
\subsubsection*{Key exchange with commuting matrices}
When the matrices commute the schemes may be simplified as follows. 

\begin{enumerate}
\item A chooses $\un{x}$, with large kernel, and $A$, non-singular, 
 and transmits $\un{x}A$.
\item B chooses $B_1$ non-singular and transmits $\un{x}AB_1$.
\item A transmits
  $\un{x}B_1$. At this stage B, and A, know $\un{x}$. A could proceed
  to transmit $\un{y}$ secretly or else:
\begin{enumerate}
\item B chooses $\un{y}$ with large kernel and $B_2$ non-singular 
and transmits $\un{y}B_2$.
\item A chooses $A_1$  non-singular  and transmits $\un{y}B_2A_1$.
\item B transmits $\un{y}A_1B$.
\end{enumerate}
\item At this stage both A and B know $\un{x}, \un{y}$ from which
  the combination is formed whose completion  is non-singular; this is 
  used as key for transmission(s) between A and B.
\end{enumerate}

%Another scheme for key exchange is the following. 
\subsubsection*{Variations} Still using the main ideas, it is clear that
many variations on the above schemes can be developed.  
For example the $\un{x}$ as $X$ can be `protected' on both
sides by transmitting $A_1XA_2$ at 1.\ above and $B_1A_1XA_2B_2$ at
2.\  where $A_iB_i=B_iA_i$ for $i=1,2$; similarly the
$\un{y}$ can be `protected' on both sides.  Methods using series of vectors
$\{\un{x_1},\un{x_2}, \ldots, \un{x_r}\}$ and $\{\un{y_1},\un{y_2},
\ldots, \un{y_r}\} $ 
are presented in section \ref{keyt}. 
%Authentication and signature as developed in section
%\ref{auth} may be introduced to the key exchange processes. 

\input{publickey}
\section{Cryptography + error-correction}\label{coding} The
cryptographic systems may be used simultaneously with 
error-correcting systems. A basic general reference for coding theory is
\cite{blahut}. 
%Embedding a coding system can ensure
%the transmitted data $\un{x}$ is such that its completion $X$ has
%even larger kernel.

Let $\un{x_1}$ be $1\ti r$ data to be transmitted securely (with
encryption)  and safely (with error coding) from A to B. 
Let $G$ be a generator $r\ti n$ matrix of an error-correcting code and
 $\un{x}= \un{x_1}G$. 
%\subsection{}
%A has data $\un{x_1}$ of size $1\ti r$ 
%to be transmitted safely and securely to B.
When matrices don't necessarily commute proceed as follows:
\begin{enumerate}
\item A works out $\un{x}=\un{x_1}G$ chooses $A$ non-singular and transmits
  $\un{x}A$. %The completion of $\un{x}A$ is $XA$.
\item B chooses $B$ non-singular and transmits $BXA$. 
\item A transmits $BX$.
\item B calculates $B^{-1}BX = X$ to get $\un{x}$
  which may have errors in transmission.  B decodes the obtained $\un{x}$ to get $\un{x_1}$. 
\end{enumerate}

If using
  $RG$-matrices of the same type only the first row of matrices need
  be worked out. 

In general it is shown in Proposition 
\ref{prop7} that if $G$ is the generator matrix of an
$(n,r)$ code which has rank $r$ and $G$ is a zero-divisor code (as are
cyclic and similar codes, see \cite{hur1}) then
the completion of $\un{x}=\un{x_1}G$ has rank at most $r$ and so $\dim
\ker $ of the completion of $\un{x}$ is $\geq (n-r)$. 

When encryption/decryption matrices to be chosen  commute the
following simplified method may be used:
\begin{enumerate}
\item   A  works out $\un{x}=\un{x_1}G$, chooses $A$ non-singular and transmits
  $\un{x}A$.
\item B chooses $B$ non-singular and transmits $\un{x}AB$.
\item A transmits $\un{x}B$.
\item B works out $\un{x}$ which may have errors in the transmissions
  and decodes to  $\un{x_1}$. 
\end{enumerate}
The code determined by $G$ is an $(n,r)$ code with rank $r$. When for
example $G$ is cyclic then $G$ can be taken as the first $r$ rows of a
circulant matrix which has rank $r$. Then the completion of $\un{x}=
\un{x_1}G$ is a circulant matrix of rank at most $r$. The kernel then
of  this completion is of dimension at least $(n-r)$. See section
\ref{codingas} for details on these aspects.

\subsection{Key exchange with coding} Modify the methods of section
\ref{key} as follows to include error-correcting codes.

Let $\{\un{x_1}, \un{y_1}\}$ be $1\ti r$ vectors so that $\{X_1,Y_1\}$ are
singular with large 
kernel and  some combination of
$\{X_1,Y_1\}$ with a known element or elements 
%(as for example the identity or half-way element)
 is non-singular. Methods for randomly being able to choose such vectors
 $\{\un{x_1},\un{y_1}\}$ are discussed in section \ref{theory} below.
 
Let $G,L$ be generator $r\ti n$ matrices of $(n,r)$ error-correcting
codes.  
\begin{enumerate}
\item A chooses $A$ non-singular and $\un{x_1}$ and transmits $\un{x_1}GA$.
\item  B chooses $B$ and transmits $BXA$ where $X$ is completion of
  $\un{x}=\un{x_1}G$.
\item A transmits $BX$.
\item B now knows $\un{x}$ which may contain errors but is decoded to
  $\un{x_1}$. %A can now repeat the process to get $Y$ secretly to B. Or else:
\begin{enumerate}
\item B chooses $B_1$ non-singular and $\un{y_1}$  
so that the completion of a combination
  of $\{\un{x_1},\un{y_1}\}$   with a known element or known elements
  is non-singular and transmits $B_1Y$ where $\un{y}=\un{y_1}L$.
\item A chooses A and transmits $B_1YA$.
\item B transmits $YA$. A knows $\un{y}$ with possible errors and
  decodes to $\un{y_1}$.
\end{enumerate}   
\item Both A and B now have $\un{x_1},\un{y_1} $ from which to form
  the encoding matrix as in section \ref{key}.
\end{enumerate}

\section{Multiple vector design}\label{mult}
The data to be transmitted is broken as $(\un{x_1},\un{x_2},
\ldots, \un{x_r})$. The $\un{x_i}$ need not be of the same length and
%are  The $\un{x}_i$ 
are arranged so that the $X_i$ are singular except for possibly a
relatively very small number of these. 
%The $\un{x_i}$ may be chosen for example so that all
%but one of the $X_i$ are singular and then the sum of the $X_i$ is
%non-singular. 

\subsection{General schemes}
$B_i$ and $A_i$ are group ring matrices and $X_i$ and $A_i$ are of the same type.
\begin{enumerate}
\item A chooses $\{A_1,A_2,\ldots, A_r\}$ non-singular and 
transmits $(\un{x_1}A_1, \un{x_2}A_2, \ldots, \un{x_r}A_r)$
\item B chooses $\{B_1,B_2,\ldots, B_r\}$ non-singular and transmits $(B_1X_1A_1,
  B_2X_2A_2, \ldots, B_rX_rA_r)$.
\item A transmits $(B_1X_1, B_2X_2, \ldots,
  B_rX_r)$. 
\item B reads 
  $(\un{x_1},\un{x_2}, \ldots, \un{x_r})$ as the first row of $({X_1},
  {X_2}, \ldots,{X_r})$.

\end{enumerate}

The matrices do not need to commute %$B_i,A_i$ need not commute
%,X_*$ need not commute 
and $B_i$ need not be of the same type  as 
$A_i,X_i$. % need not be derived from same group rings.
If $B_i$ is of the same type as $X_i,A_i$ then only the first rows of
the matrices need be transmitted in 2.\, 3.\ above. In these cases
convolution methods for multiplication may be used.

\subsubsection{Schemes with some matrices commuting}
Here we have matrices $A_i,B_i$ with $A_iB_i=B_iA_i$ for each $i$; 
it is not necessary that
$A_iB_j=B_jA_i$ for $i\neq j$ nor that $A_iA_j=A_jA_i$ for any $i,j$.
\begin{enumerate}
\item A chooses $\{A_1,A_2,\ldots, A_r\}$ non-singular and 
transmits $(\un{x_1}A_1, \un{x_2}A_2, \ldots, \un{x_r}A_r)$
\item B chooses $\{B_1,B_2,\ldots, B_r\}$ non-singular and transmits
  $(\un{x_1}A_1B_1, \un{x_2}A_2B_2, \ldots,  \un{x_r}A_rB_r)$.
\item A transmits  $(\un{x_1}B_1, \un{x_2}B_2, \ldots,
  \un{x_r}B_r)$.
\item B reads  $(\un{x_1}, \un{x_2}, \ldots,
  \un{x_r})$.
\end{enumerate}

\subsection{Key exchange with multiple vectors and matrices}\label{keyt}
 Key exchange with multiple vector choices may be achieved as follows:

Let $\{\un{x_1},\un{x_2},\ldots, \un{x_r}\}$ and  $\{\un{y_1},\un{y_2},
  \ldots, \un{y_r}\}$ be sets of vectors where for each $i$,
  $\un{x_i}$ has the same length as $\un{y_i}$; these are chosen
  randomly so
  that  $X_i,Y_i$ are singular  (except possibly for a relatively small
number of them)  and 
 some combination of $X_i,Y_i$ is non-singular or some combination of
 $X_i,Y_i$ with a known element or known elements  is non-singular. 
\begin{enumerate}
\item A chooses $\{A_1,A_2,\ldots,A_r\}$ non-singular and
  $(\un{x_1},\un{x_2}, \ldots,\un{x_r})$ and 
transmits $(\un{x_1}A_1,\un{x_2}A_2,\ldots, \un{x_r}A_r)$.
\item  B chooses $\{B_1,B_2,\ldots, B_r\}$ non-singular and 
transmits $(B_1X_1A_1,B_2X_2A_2, \ldots, B_rX_rA_r)$.
\item A transmits $(B_1X_1,B_2X_2,\ldots, B_rX_r)$.
\item B now knows $(X_1,X_2,\ldots, X_r)$. 
A can now repeat the process to get $(Y_1,Y_2,\ldots, Y_r)$
  secretly to B. Or else:
\begin{enumerate}
\item B chooses $(\un{y_1},\un{y_2},\ldots, \un{y_r})$ so that 
 a combination of $X_i,Y_i$ or a combination of $X_i,Y_i$ with a
known element or known elements is  non-singular. 
\item B chooses $\{B_1',B_2',\ldots, B_r'\} $ non-singular and transmits
  $(B_1'Y_1,B_2'Y_2,\ldots, B_r'Y_r)$.
\item A chooses $\{A_1',A_2',\ldots,A_r'\}$ non-singular and transmits
  $(B_1Y_1A_{1}',B_2Y_2A_{2}',\ldots, B_rY_rA_{r}')$. 
\item B transmits $(Y_1A_{1}',Y_2A_{2}',\ldots, Y_rA_{r}')$. 
\end{enumerate}   
\item Both A and B now have the  $X_i,Y_i$ for each $i$ 
from which to form the secret encryption matrices.
\end{enumerate}

\subsection{Key exchange with multiple vectors and coding}
Key exchange with multiple vector choices and coding may be achieved as follows:

Let $\{\un{\un{x_1}},\un{\un{x_2}},\ldots, \un{\un{x_r}\}} $ and
$\{\un{\un{y_1}},\un{\un{y_2}}, 
  \ldots, \un{\un{y_r}}\}$ be sets of vectors where for each $i$,
  $\un{\un{x_i}}$ has the same length as $\un{\un{y_i}}$; these are
  chosen randomly so
  that  their completions $\un{X_i},\un{Y_i}$ are singular except
  possibly for a small number of them and 
 some combination of their completions is non-singular or
 some combination of the 
 completions with known elements are non-singular.  Define
 $\un{x_i}=\un{\un{x_i}}G_i, \un{y_i}=\un{\un{y_i}}K_i$ for appropriately
 sized generator matrices $G_i,K_i$ of error-correcting codes.
\begin{enumerate}
\item A chooses $\{A_1,A_2,\ldots,A_r\}$ non-singular and 
transmits $(\un{x_1}A_1,\un{x_2}A_2,\ldots, \un{x_r}A_r)$.
\item  B chooses $\{B_1,B_2,\ldots, B_r\}$ non-singular and 
transmits $(B_1X_1A_1,B_2X_2A_2, \ldots, B_rX_rA_r)$.
\item A transmits $(B_1X_1,B_2X_2,\ldots, B_rX_r)$.
\item B now knows $(\un{x_1},\un{x_2},\ldots,\un{x_r})$ with possible errors and
  decodes this  to $(\un{\un{x_1}}, \un{\un{x_2}}, \ldots, \un{\un{x_r}})$.
A can now repeat the process to get
$(\un{\un{y_1}},\un{\un{y_2}},\ldots, \un{\un{y_r}})$ 
  secretly to B. Or else:
\begin{enumerate}
\item B chooses $(\un{\un{y_1}},\un{\un{y_2}},\ldots, \un{\un{y_r}})$ so that 
 a combination of $\un{X_i},\un{Y_i}$ or a combination of $\un{X_i},\un{Y_i}$ with a
known element or elements is  non-singular. 
\item B chooses $\{B_1',B_2',\ldots, B_r'\} $ non-singular and transmits
  $(B_1'\un{y_1},B_2'\un{y_2},\ldots, B_r'\un{y_r})$.
\item A chooses $\{A_1',A_2',\ldots,A_r'\}$ non-singular and transmits
  $(B_1Y_1A_{1}',B_2Y_2A_{2}',\ldots, B_rY_rA_{r}')$. 
\item B transmits $(Y_1A_{1}',Y_2A_{2}',\ldots, Y_rA_{r}')$. A then
  knows $(\un{y_1},\un{y_2},\ldots, \un{y_r})$ with possible errors and decodes  to
  $\un{\un{y_1}},\un{\un{y_2}},\ldots, \un{\un{y_r}})$. 
\end{enumerate}   
\item Both A and B now have the  $\un{\un{x_i}},\un{\un{y_i}}$ for each $i$ 
from which to form the secret encryption matrices.
\end{enumerate}

\section{Who is there?}\label{who}
Authentication and/or signature methods may be set up in the usual
way when key exchange and/or public key schemes have been
established. Section \ref{define} shows how public key may in a unique
way be used to establish that message is actually emanating from a
correspondent A; the constituents of the public key for B are changed
so the new key may be used only by a particular A.  

Without using public key or key exchange 
one or both of the following may be requirements.
\begin{itemize}
\item In a message exchange from A to B  it may be the case that a response 
 from B is
required. In certain situations then A requires to know
that no one else is responding pretending to be B. (`Person-in-middle'
problem.)  
 
\item B requires to know that 
message purporting to come from A is actually from A. 
\end{itemize}
%Authentication and signature methods with these in mind are set
%up. 
To prevent these `person-in-middle' problems proceed as follows. 
%enable these one of the participants at least 
%must  have `signature' which needs to
%be trusted; this can be fixed or ephemeral or indeed vary continuously. 
Each person X must have a `key' which is of the form $\un{y_X}X$ where
$\un{y_X}$ has large kernel. This must be known  to and trusted by 
the person with whom
the contact is to be made but may be public. This is a `partial'
public key as discussed in   section \ref{define1}. 

  %as outlined in next section \ref{auth}.

\subsection{Prevent Eve pretending to be B}\label{auth}
%Can we get over issue of `pretend' by using signatures?
E(ve), an eavesdropper, looking in at the communications in section
\ref{tryy} or \ref{encrypt1}   
can see $\un{x}A$ and pretends
to be B.  (S)he then applies $E$ to get 
$EXA$ in \ref{tryy} or $\un{x}AE$ in \ref{encrypt1},  which is then transmitted to
A who applies $A^{-1}$ and gives back $EX$ or $\un{x}E$ to E who can then read
off $\un{x}$.

%Each person makes up a `signature' which is revealed as required; it
%may change periodically. It could be part of a key exchange.   
 %this can be ephemeral. %change.  
A wants to communicate $\un{x}$ secretly to B. As stated the key  for B 
is constructed from a vector $\un{y}$ and a non-singular matrix
$B$ and only the product $\un{y}B$ is known to A but it may be public.
 %This `signature' should be trusted.  
%Let  $\un{y}$ be the signature
%of B (at this time). 
$\un{y}$ should be chosen
so that its completion $Y$ is singular and 
has large  kernel. However $\un{y}B$ is not a public key for B in
general as it does not have an inverse.
% in any case itis best that 

Use the convention that matrices $A$ and $A_*$  for suffices $*$ are
matrices chosen and applied by A(lice) and $B$ and $B_*$ are matrices
chosen and applied by B(ob). 
\subsubsection{With Commuting matrices}\label{comm2} Suppose the
matrices commute.  
%To avoid A knowing $\un{y}$ (as at item 4.\, above) at all proceed as follows:
\begin{enumerate}
\item B chooses signature key $\un{y}B$ which is revealed.
 %at a particular  time. 
\item A chooses $\{A,A_1\}$ and 
sends out $(\un{x}A,\un{y}BA_1)$.
\item B chooses $\{B_1,B_2\}$ and transmits $(\un{x}AB_1,\un{y}A_1B_2)$.
\item A works out $%(\un{x}AB_1A^{-1},\un{y}A_1B_2A_1^{-1})=
  (\un{x}B_1,\un{y}B_2)$ and  transmits $(\un{x}B_1 - \un{y}B_2)$. 
\item B works out $(\un{x}-\un{y}B_2B_1^{-1})$ 
 and  $\un{y}B_2B_1^{-1}$ and adds the two to get 
  $\un{x}$. 
\end{enumerate} 

In fact for 5.\ $\un{y}B_2B_1^{-1}$ can be worked out when $B_1,B_2$ are
  chosen at 3.\ .
Now A never knows $\un{y}$ in this set-up so B may use the same
$\un{y}B$ in communicating with another.

At point 4.\ A knows $\un{y}B_2$ and may use this it in further
transactions from A to B avoiding some transmissions at points 2.\ , 3.\ 
above. In a sense then when A knows $\un{y}B_2$ it may be as a `key'
for transmissions from A to B and may be used as a non-public signature
of B for A only. Some simplification can be initiated when B is not
worried that A may find $\un{y}$.  
\subsubsection{With matrices which may not commute}\label{sub5}

%When $RG$-matrices are used the following may be adopted and the
%matrices need not necessarily commute.

%\subsubsection{Pretending, authentication}\label{sub3}

Similar schemes using non-commuting matrices are developed as follows.

A is required to transmit $\un{x}$ to B and make sure that an
eavesdropper may not pretend to be B. The matrices $A_*$ and $B_*$
need not be of the same type, that is, need not be formed from the 
same group ring.

%A is required to transmit $\un{x}$ to B and make sure that an
%eavesdropper may not pretend to be B. The matrices $A_*$ and $B_*$
%need not be of the same type, that is, need not be formed from the 
%same group ring.
\begin{enumerate}
\item B chooses $\un{y}$ and $B$ and circulates $\un{y}B$ (keeping
  $\un{y}$ and $B$ secret).  
\item A chooses $A,A_1$ and transmits $(\un{x}A, A_1YB)$.
%\item B completes $\un{x}A$ to $XA$.
\item B chooses $B_1,B_2$ and transmits $(B_1XA,A_1YB_2)$
 \item A works out $(B_1X,YB_2)$
and transmits $(B_1X-YB_2)$.
\item B works out $B_1^{-1}(B_1X -YB_2)= X-B_1^{-1}YB_2$ and adds it to
  $B_1^{-1}YB_2$, which may be worked out previously, to get $X$.
\end{enumerate}

At point 4.\ A knows $YB_2$ which can be used for further transactions
from A to B.
 
Variations on the above are easily constructed and designed.

\subsection{To be sure}\label{sub2}
Suppose now A communicate with B and B wishes to be sure 
that the message is from A.

\subsubsection{Where from,  commuting} 
Each participant X has $\un{y_X}X$, where $X$ is
invertible. When  $RG$-matrices are used the completion of $\un{y_X}$
should be singular with large kernel. $\un{y_X}$ and $X$ are kept
secret. %Suppose the matrices commute. 
\begin{enumerate} \item A chooses $A_1$ and transmits $\un{y_A}A_1$.
\item B chooses $B_1$ and transmits $\un{y_A}A_1B_1$. 
\item A transmits $\un{y_A}AB_1$ and B checks this. % to see it comes
                                % from A.
\end{enumerate}

(At stage 2.\ B can work out $\un{y_A}AB_1$ for checking at 3.\.)

\subsubsection{Where from, non-commuting}\label{sub4}

Suppose A wishes to communicate with B and B wishes to be sure 
that the message is from A.
Each participant X publishes $XY_X$, where $X$ is
invertible and $Y_X$ is singular with large kernel. $Y_X$ and $X$ are kept
secret.  
\begin{enumerate} \item A chooses $A_1$ and transmits $A_1Y_A$.
\item B chooses $B_1$ and transmits $A_1Y_AB_1$. (At this stage B can work
  out $AY_AB_1$.)
\item A transmits $AY_AB_1$ and B checks this. % to see it comes from A.
\end{enumerate}

\subsection{Combined}\label{comb}
The methods of \ref{auth}, \ref{sub2} may be combined as required or necessary. 
A wishes to communicate with B; A requires that an eavesdropper may
not pretend to be B and B requires a signature so that (s)he knows the
message is from A. The methods are fairly straightforward and  
details are omitted. 

%\subsubsection{Combined, non-commuting}
%A wishes to communicate with B; A requires that an eavesdropper may
%not pretend to be B and B wishes to be sure that the message is from
%A.
%This is similar to the above section \ref{sub7} for commuting matrices case by
%combining the methods in
%sections \ref{sub4} and \ref{sub5} for the non-commuting cases;
%details are omitted. 

%The $\un{y}$ that B chooses in 1.\ should be arranged to be singular with
%large kernel. 

\subsection{Multiple vector design: Prevention} The authentication, signature
methods devised above may also be extended to multiple vector
design. We outline  just one of the methods.
\subsubsection{Prevent E pretending}
It is required when A communicates with B that E may not reply to
A succeeding in pretending to be B.
\begin{enumerate}
\item B has a  key $(\un{y_1}B_1,\un{y_2}B_2, \ldots,
  \un{y_s}B_r)$ which is revealed at a particular time and known and
  trusted by A. 
\item A sends out $((\un{x_1}A_1,\un{x_2}A_2,\ldots,
  \un{x_r}A_r),(\un{y_1}B_1A_{1}',
  \un{y_2}B_2A_{2}',\ldots,\un{y_s}B_rA_{r}'))$ where
  $(\un{x_1},\un{x_2},\ldots, \un{x_r})$ is the data to
  be transmitted and the size of $x_i$ is the same as that of $y_i$.  
\item B chooses $(B_{1}',B_{2}',\ldots, B_{r}')$ and transmits 
$((\un{x_1}A_1B_{1}',\un{x_2}A_2B_{2}',\ldots, \un{x_r}A_rB_{r}'),
  (\un{y_1}A_{1}', \un{y_2}A_{2}', \ldots, \un{y_r}A_{r}'))$.
\item A %works out $(\un{y_1},\un{y_2}, \ldots, \un{y_r})$ and
  transmits $(\un{x_1}B_{1}',\un{x_2}B_{2}',\ldots, \un{x_r}B_{r}') -
  (\un{y_1},\un{y_2}, \ldots, \un{y_r})$. 
\item B works out $(\un{x_1},\un{x_2}, \ldots, \un{x_r})-
  (\un{y_1}B_{1}'^{-1}, \un{y_2}B_{2}'^{-1}, \ldots, \un{y_r}B_{r}'^{-1})$ and
  adds this to $(\un{y_1}B_{1}'^{-1}, \un{y_2}B_{2}'^{-1}, \ldots,
  \un{y_r}B_{r}'^{-1})$, which has already been worked out, to get
  $(\un{x_1},\un{x_2}, \ldots, \un{x_r})$. 
\end{enumerate} 
\subsection{Authentication, signature, + coding}
Authentication and signature with coding may similarly be
implemented. The details are omitted. Basically  first of all the data is encoded as
$\un{x}= \un{x_1}G$. Then when $\un{x}$ is received  with
possible errors it is decoded to $\un{x_1}$. 

%\input{example1}

%\subsubsection{Disadvantages}
%\begin{itemize}
%\item The necessity of three transmissions when keys are not
%  exchanged. However the speed of both
%  encryption and decryption can well make up for this.
%\item It may be necessary to introduce a `signature' to prevent
%  intruders from pretending to be someone else. The signature system
%  is very secure also. (Other crypto systems also require signatures
%  to prevent this type of `pretending'.)
%\end{itemize}  
%\input{rankresults}

\section{Theory}\label{theory}

\subsection{Vector by matrix multiplication}
Much is contained  in the literature on vector-matrix/matrix-vector 
multiplication. The multiplication  can be  very fast when the matrix has a
structure as for example if the matrix is an $RG$-matrix; a
circulant matrix is such an example. Group ring matrices  of the
groups $C_n,C_2^n, C_p^n$ and in general abelian
groups are particularly suitable. 
Vector-matrix multiplication in these cases using fast Fourier
transform or Walsh-Hadamard fast transform (for $FC_p^n$) can be done
in $0(n\log n)$ time. 

The multiplication can be done in $0(n\log n))$ time for
$FG$-matrices when $G$ is a finite supersolvable group;
this comes from Baum's
Theorem  \cite{baum} which states that every
supersolvable finite group has a DFT (Discrete Fourier Transform) 
algorithm running in $O(n\log n)$
time. This more general notion is not discussed further here  % on this see 

\subsection{Rank and nullity}
 Knowledge of a singular matrix and a
product of this singular matrix by a non-singular matrix does not lead
to knowledge of the non-singular matrix. It is desirable that the
kernel 
of the singular matrix be relatively large. 
%This is easy for relatively large size $n\ti n$ matrices. 
 The nullity  of an $n\ti n$  matrix with  $A^t=0$ is greater than 
 or equal to $\frac{n}{t}$; see Corollary \ref{rank1} below. 
For large $n$ and relatively small $t$
 a solution of a system of equations as $AX=B$, with $X$ as
 indeterminates or $\un{x}A=\un{b}$ with indeterminates $\un{x}$ 
then has many possible solutions.    
If
 $A$ has (relatively) small rank then so does $AY$ and $YA$ for any $Y$ as 
 $YA\leq \min \{\rank X, \rank A\}$ and $AY\leq \min \{\rank, Y, \rank
A$\}. 
%Thus there are many possible
% solutions to systems such as  $AX=B$ or $XA=B$  for an indeterminate matrix $X$
% and also many possible solutions to  $\un{x}A = \un{b}$ for an indeterminate
% vector $\un{x}$.  

%cannot lead
% to a guess of result.  
% when the rank of the singular
%matrix is not large.
\begin{lemma}\label{rankkk}
Let $A$ be an $n\ti n$ matrix such that $A^t=0$. Then $\rank A \leq
\frac{n(t-1)}{t}$.
\end{lemma}
\begin{proof}
Note first that for $n\ti n$ matrices $\rank PQ\geq \rank P + \rank Q
-n$.

Suppose then $\rank A > \frac{n(t-1)}{t}$. We now show by induction
that $\rank A^r > \frac{n(t-r)}{t}$ for $1\leq r \leq t$. The case $r=1$ is part of the
  hypothesis. Suppose then $\rank A^k > \frac{n(n-k)}{t}$ for $1\leq k
  < t$. 
Hence $\rank A^{k+1}=\rank AA^k \geq \rank A + \rank A^k -n >
 \frac{n(t-1)}{t} + \frac{n(t-k)}{t}-n = \frac{n(t-(k+1))}{t}$ as
 required.

Now  $AA^{t-1}=0$ implies that $A^{t-1}\subseteq \ker A$ and so $\rank
A^{t-1} \leq \dim \ker A$. But $\rank A + \dim \ker A = n$ implies
$\dim \ker A =
n - \rank A = n - \frac{n(t-1)}{t} = \frac{n}{t}$ and so $\rank
A^{t-1} \leq \frac{n}{t}$. However letting $r=t-1$ in $\rank A^r >
\frac{n(t-r)}{t}$ implies $\rank A^{t-1} > \frac{n}{t}$ which is a
contradiction. Hence $\rank A \leq \frac{n(t-1)}{t}$.
\end{proof} 
\begin{corollary}\label{rank1} Suppose $A^t=0$ for an $n\ti n$ matrix. Then $\dim
  \ker A \geq \frac{n}{t}$.
\end{corollary}
\begin{proof} This follows from the Lemma since $\rank + \dim \ker =
  n$. 
\end{proof}

\medskip 
The following may also be shown but is not relevant here:
\begin{lemma} Suppose $A$ is an $n\ti n$ matrix with $A^t=0$. Suppose
  also $\rank A = \frac{n(t-1)}{t}$. (This is largest it can be by
  Lemma \ref{rankkk}.). Then $\rank A^{t-1}= \frac{n}{t}$. In
  particular this implies $A^{t-1}\neq 0$. 
\end{lemma}

%(Problem: If $A^t=0$ and $A^{t-1}\neq 0$ what can $\rank A$ be?) 

\subsection{$RG$-matrices}\label{results}

An $RG$-matrix is a matrix corresponding
to a group ring element in the isomorphism from the group ring
into the ring of $R_{n\times n}$  matrices, see for example \cite{hur}.  
Specifically suppose $w = \di\sum_{i = 1}^n\al_{g_i}g_i \in RG$ where
$G=\{g_1,g_2,\ldots, g_n\}$ is a listing of the elements of $G$. 
The $RG$-matrix of $w$ denoted by $M(RG,w)$ is defined as follows:

 $\begin{pmatrix}
\alpha_{g_1^{-1}g_1} & \alpha_{g_1^{-1}g_2} &\alpha_{g_1^{-1}g_3} 
 &  \ldots & \alpha_{g_1^{-1}g_n} \\

\alpha_{g_2^{-1}g_1} & \alpha_{g_2^{-1}g_2} &\alpha_{g_2^{-1}g_3} 
 &  \ldots & \alpha_{g_2^{-1}g_n} \\ 

\vdots & \vdots & \vdots &\vdots &\vdots \\

\alpha_{g_n^{-1}g_1} & \alpha_{g_n^{-1}g_2} &\alpha_{g_n^{-1}g_3} 
 &  \ldots & \alpha_{g_n^{-1}g_n} 
%\ldots & \ldots & \ldots & \alpha_0 & \ldots 
\end{pmatrix} $

The matrix is in $R_{n\times n}$ and 
 depends on the listing of the elements. 
Changing the listing changes the matrix;  if $A,B$ are $RG$-matrices
for the element $w\in RG$ relative to different listings then $B$ may be
obtained from $A$ by a sequence of [interchanging two rows and
 then interchanging the corresponding two columns].  

Given the entries of the first row of an  $RG$-matrix, and a listing,  
the entries of the other rows are determined from the multiplication
of the elements of $G$ and 
each row and each column is a permutation of the first row.
%determined by the matrix of $G$.

%\newtheorem{groupring}{Theorem}

\begin{theorem}
Given a listing of the elements of a group $G$
  of order $n$ there
is a bijective ring homomorphism between $RG$ and the
$n\times n$ \, $RG$-matrices. This bijective ring homomorphism is given
  by $\sigma: w \mapsto M(RG,w)$.
 \end{theorem}

An $RG$-matrix for a cyclic group $G$ is a circulant matrix; 
an $RG$-matrix when $G$ is a dihedral group 
is one of the form $\begin{ssmatrix} A&B
\\ B&A\end{ssmatrix}$  (in a natural listing of the elements of $G$), 
where $A$ is circulant and $B$ is reverse circulant. 

For $w\in RG$ the corresponding capital letter
$W$ denotes the image of $w$ in the ring of $R_{n\ti n}$ matrices,
relative of course to a particular listing of the elements of $G$. For 
a vector $\un{x}\in R^n$ and a fixed listing of a group $G$ 
by convention the capital letter $X$, without underlining, denotes the
completion of $\un{x}$.  
 Say $w\in RG$ is singular
if and only if $W\in R_{n\ti n}$ is a singular matrix and $w$ is
non-singular if and only if $W$ is a non-singular matrix. Thus when
$R$ is a field $w$ is
singular if and only if $w$ is a zero-divisor in $RG$, and $w$ is non-singular
if and only if $w$ is a unit in $RG$, \cite{hur}. 

%Further details are given
%below where it is shown how it may be arranged that the kernel of $X$
%may be arranged to be large. %element and matrix singular.  
%Then  then %The kernel of $X$ should also be large as can be arranged, see below.

\subsection{Commuting matrices}\label{comm}
Matrices that commute with one another include group ring matrices
corresponding to group rings of abelian groups. Convenient such group ring
matrices include:
\begin{enumerate}
\item Circulant matrices over any field; in particular 
 circulant matrices over finite fields such as $Z_p$ for $p$ a prime.  
\item $RG$-matrices from $RG=Z_2C_2^{n}$. An element $w=\sum_{i=0}^{2^n-1}\al_ia_i$
  is invertible if and only if $\sum_{i=0}^{2^n-1}\al_i=1$, that is, if and
  only if there are an odd number of non-zero coefficients in $w$. For say
  $n=1024$ there are $2^{1023}$ such invertible elements and
  $2^{1023}$ elements whose square is zero.
\item Matrices from   $Z_pC_p^{n}$. Let 
  $w=\sum_{i=0}^{p^n-1}\al_ia_i\in \Z_pC_p^n$, with $\al_i\in \Z_p,
  a_i\in C_p^n$. Since $w^p=\sum \al_i$, it follows that $w$ 
  is invertible if and only if $\sum_{i=0}^{p^n-1}\al_i \neq 0$. 
 If this sum $s=\sum \al_i$ is zero then $w$ is a
  zero-divisor with $w^p=0$ and if this sum $s\neq 0$ then $w^{-1}=s^{-1}w^{p-1}$.

For say
  $n=102$ there are $p^{101}(p-1)$ such invertible elements and
$p^{101}-1$ such non-zero elements which are zero-divisors satisfying
$w^p=0$. It is easy to choose  randomly an 
invertible element whose inverse is easy to construct or a zero-divisor element
 with relatively small power equal to zero.
%Each element $w \in \Z_pC_p^n$ and corresponding $W$ in the matrix
%ring satisfy $w^p=\al 1 $, $W^p=\al I $ for $\al$ a scalar.  
\end{enumerate} 
The types of matrices used for the designs and for the
transmissions of vectors  need not be the same.

%In order to prevent an attack by taking powers  make sure 
%that the completion $X$ of $\un{x}$ satisfies 
%$X^t=0$ for some relatively small $t$;  no information can then 
%be obtained by taking
%powers of $AX$ for any revealed $AX$. 

%If $w\in \Z_2C_2^n$ then $w^2=1$ or
%$w^2=0$; in case $w^2=1$
%adding one extra term to $w$  will ensure the new $w_1$ satisfies $w_1^2=0$. 

\subsection{Construction methods}\label{singular} For our constructions it is
required to randomly choose, from a large available pool, matrices and
vectors of the following types:
\begin{itemize}
\item Singular matrices $A$ with large kernel.
\item Non-singular matrices $A$ such that  the inverse of
  $A$ is easy to compute;
\item Vectors  $\un{x},\un{y}$ such that $X,Y$ have large kernels and
  a combination of $X,Y$ with a known element or known elements is
  non-singular, the inverse of which is easy to obtain.

\end{itemize} 

Further:
\begin{itemize}
\item Given data $\un{x}$ it is required to construct
$\un{\un{x}}$ from which $\un{x}$ may directly be obtained and for
  which the completion of 
$\un{\un{x}}$ is singular with large kernel.
\end{itemize}

Here we show how such constructions may be obtained in various group ring matrices.

\subsubsection{In $\Z_2C_2^n$}
Consider $\Z_2C_2^n$. % initially.
 An element
$w=\sum_{i=0}^{2^n-1}\al_ia_i \in \Z_2C_2^n$ satisfies
$w^2=\sum_{i=0}^{2^n-1} \al_i$ and so $w^2=0$ or $w^2=1$ according to
whether the sum of the coefficients of $w$ is even or odd. When
the sum is even then $w^2=0$ and so indeed $w$ is singular with large
kernel by Corollary \ref{rank1}. 
In $\Z_2C_2^n$ it is easy to arrange for any data $\un{x}$ that
if $\un{x}^2\neq 0$ then  adding one known element to $\un{x}$ ensures
the square of the data is zero. When $\un{x}^2=0$ then $\rank X$,
where $x$ is the completion of $\un{x}$ is at most $\frac{n}{2}$ and
thus $\dim \ker X \geq \frac{n}{2}$; for large $n$ this ensures $\dim
\ker X$ is large. Thus there are at least $2^{\frac{n}{2}}$ solutions
in $\un{z}$ to $X\un{z}=\un{b}\T$ or $XZ=P$ for unknown matrix $Z$. 

Thus in $\Z_2C_2^n$: %it is easy to obtain construct as follows:
\begin{itemize}
\item Random Matrices $X$ may be chosen such that $X^2 =0$ and so has
  large kernel; 
\item Random Matrices $A$ may be chosen 
such that $A^2=1$ and so the inverse is easy to  obtain. 
\item Random $\un{x},\un{y}$ may be chosen so that $X^2=0, Y^2=0$ and then
  both $X,Y$ are singular with large 
  kernel. Combinations such as $X+Y+1$, $X+Y+H$ where $h\in C_2^n$
  and  $X+Y+w$ where 
  $w$ has an odd number of non-zero terms have their squares equal to $1$. 
\item If $\un{x}$ is any vector considered in $\Z C_2^n$ then either $X^2=0$
  and has large kernel or else adding an element $h$ of $C_2^n$ ($h$
  could be the identity)
  ensures $(X+H)^2=0$, or more generally adding an element $w$ with an
  odd number of non-zero terms ensures $(X+W)^2=1$.
\end{itemize} 

%Group ring matrices from  $C_2^n$-matrices are particularly
%interesting for our purposes.
%; these are matrices produced from the
%group ring of $C_2^n$. 
For any ring $R$, an $RC_2$ matrix is
one of the form $\begin{ssmatrix} \al & \be \\ \be &
  \al \end{ssmatrix}$ with $\al, \be \in R$.  An $RC_2^n$ matrix for $n\geq 2$ 
is  one of the form $\begin{pmatrix} A_{n-1} & B_{n-1} \\ B_{n-1} &
  A_{n-1} \end{pmatrix}$ where $A_{n-1},B_{n-1}$ are
$RC_2^{n-1}$-matrices. An $RC_2^n$-matrix is completely determined by
its first row as is any $RG$-matrix.

Any  $RC_2^n$-matrix is diagonalised by the
Walsh-Hadamard $2^n\ti 2^n$ matrix which is defined as follows.
The Walsh-Hadamard $2\ti 2$ matrix is $W_2=\begin{ssmatrix} 1 & 1 \\ 1 &
  -1 \end{ssmatrix}$ and for $n\geq 2$ the Walsh-Hadamard $2^n\ti 2^n$
matrix is $W_{2^n} =\begin{pmatrix} W_{2^{n-1}} & W_{2^{n-1}}
  \\  W_{2^{n-1}} & -W_{2^{n-1}}\end{pmatrix}= W_2\otimes W_{2^{n-1}}$
where $\otimes$ denotes tensor product. 
It is known that the Walsh-Hadamard transformation can be
performed in time $O(m\log m)$ ($m=2^n$) and thus vector and matrix operations
with $RC_2^n$-matrices can be done in $O(m\log m)$ time.   
Thus using $\Z_2C_2^n$ the constructions may be done in $O(m\log m)$
time using Walsh-Hadamard transformations. 

\subsubsection{In $\Z_pC_p^n$}
Consider now the data $\un{x}=(\al_0,\al_1,\ldots,\al_{p^{n}-1})$ to be
 in $\Z_pC_p^n$, that is $\un{x} =\di\sum_{i=0}^{p^n-1}\al_ig_i$ where
\\ $\{g_0,g_1,\ldots, g_{p^n-1}\}$ are the elements of $\Z_p^n$ and $\al_i\in
Z_p$. Each $g_i$ satisfies $g_i^p=1$. Then
$\un{x}^p=\di\sum_{i=0}^{p^n-1}\al_i^{p}=\sum_{i=0}^{p^n-1}
\al_i=\ep(\un{x})$ where $\ep(\un{x})$ denotes the augmentation of
$\un{x}$.
If $\ep(\un{x})=0$ then $\un{x}^p=0$. If $\ep(\un{x})\neq 0$
then $(\un{x}-\ep(\un{x})g)^p=0$ for any $g\in C_p^n$. More generally
$(\un{x}+\sum_{j\in J}b_jg_j)^p=0$ when $\sum_{j\in
  J}b_j=-\ep(\un{x})$ for $J\subset \{0,1,\ldots, p^n-1\}$. 

Thus it is easily arranged for the data $\un{x}$ to satisfy
$\un{x}^p=0$ by adding a known element or known elements as necessary. If now
$\un{x}^p=0$ then the completion $X$ of $\un{x}$ has $\dim \ker X \geq
\frac{n}{p}$. Hence any system of equations $X\un{z}=\un{b}$ for
 unknown $\un{z}$ has $p^{\frac{n}{p}}$ solutions.

In $C_{p}^n$ every element has order $p$ so $w=
\sum_{i=0}^{2^n-1}\al_ig_i\in \Z_pC_p^n$ satisfies
$w^p=\sum_{i=0}^{2^n-1}\al$. When $w^p \neq 0$ then $w^p=\ep(w)\neq 0$
and the inverse of $w$ is easy to obtain. %then adding $\sum
%\al_i$ to $w$ gives a singular element with $p^{th}$ power equal to
%$0$. When the $p^{th}$ power is $0$ then by Lemma \ref{rankkk} $\dim \ker w
%\geq \frac{n}{p}$ which is large when $n$ is large compared to
%$p$. 
%If $w=\un{x}\in \Z_p^n$ is data to be transmitted then as an
%element in $\Z_pC_p^n$ either $w^p =0 $ or else $w^p=\al \neq 0$
%($\al$ is the sum of the coefficients) and then $w'=w - \al g$ for any
%$g\in C_p^n$ satisfies $w'^p=0$; thus data may be made with large
%$\dim \ker$ by adding just one term. 

Thus in $\Z_pC_p^n$: %it is easy to obtain:
\begin{itemize} \item Random matrices $X$ may be chosen such that
  $X^p=0$ and so $X$ has  large kernel. 
\item Random matrices $A$ may be chosen  such that $A^p=\al I$ for a
  scalar $\al$ and  hence the inverse of $A$ is easily obtained.
\item It is possible to randomly choose $X,Y$ so that $X^p=0,Y^p=0$
  and so $X,Y$ have large kernel and $(X+Y+1)^p$ or $(X+Y+H)^p$ with
  $h\in C_p^n$ or $(X+Y+W)^p$ for various $w\in C_p^n$ to have value
  $\al I$ for a scalar $\al$. 
\item If $\un{x}$ is any vector considered in $\Z_pC_p^n$ then either $X^p=0$
  or else $(X+H)^p=I$ for any $h\in C_p^n$; more generally $X^p=0$ or
  $(X+W)^p=0 $ for  $\ep(w)=-\ep(\un{x})$. 
\end{itemize}

A generalised Walsh-Hadamard matrix $WH(p^n)$ is defined as
follows.  
$WH(p)=F_p$ where $F_p$ is the Fourier $p\ti p$ matrix and $WH(p^n)=
WH(p) \otimes WH(p^{n-1})$ for $n\geq 2$ where $\otimes$ denotes tensor
product. This diagonalises any
$RC_p^n$-matrix when the Fourier matrix exists. 
Using generalised Walsh-Hadamard matrices
computations in $\Z_pC_p^n$ can be done in $O(m\log m)$ time, $m=p^n$.    

\subsubsection{With circulants}
%Any $R 
%\input{RG}
%\subsubsection{Group ring designs with required properties}
Suppose circulant matrices derived
from $\Z_2C_n$ are used where $n=2m$ is large. Let  $C_n$ be generated by
$a$. Let $J\subset \{0,1,\ldots,m\}$ be  chosen randomly. 
Now $w= (\sum_{j\in J}
(a^j+a^{m+j}))+a^m$ satisfies $w^2=(\sum_{j\in J}(a^{2j}+ a^{2m+2j})) +
a^{2m}= (\sum_{j\in J}(a^{2j} + a^{2j})) + 1 = 1$. (One could also use 
 $w= (\sum_{j\in J}
(a^j+a^{m+j}))+1$.) The circulant matrix
$W$ which is the completion of $w$ satisfies $W^2=1$. The number of choices for
such $J$ is of order $2^m$. 
Then $A,B$ above can then be constructed from choices of $J$. 
%For $m=1024$ say there is no way one
%could find the matrix $A$ or $B$ or $B_1$.   
%In $\Z_2C_2^n$ every element $w$ satisfies $w^2=1$ or $w^2=0$ depending
%on whether the of $w$ is odd or even. Thus $A,B$ could be chosen
%randomly from $\Z_2C_2^n$ so that $A^2=1=B^2$; in this case also it is
%easy to ensure the data $\un{x}$ is such that $X^2=0$ and so, by
%Lemma \ref{rankkk}, that $\dim \ker AX \geq n/2, \dim \ker BX \geq n/2$. 
%Similarly in $\Z_pC_p^n$ and $\Z_pC_{pm}$ 
%random non-singular elements may be chosen 
% and the data made to have small power
%equal to $0$. 

Consider $\Z_pC_{pn}$ where $C_{pn}$ is generated by $a$. It is easy
to build singular elements $w$ with $w^p=0$. Now $(a^i+(p-1)a^{i+n})^p 
=a^i+(p-1)a^i=0$ and also $(a^i+a^{i+n} + \ldots + a^{i+(p-1)n})^p= 0$
and other similar constructions. Taking a sum of such types gives an
element $w$ with $w^p=0$ whose completion is a singular element
with $\dim \ker \geq \frac{n}{p}$. 

Matrix and vector multiplication for circulant matrices ($RG$-matrices
for $G$ cyclic) can be done
with fast Fourier transform and so can be done in $O(n\log n)$ time.   

Thus in $\Z_pC_{pm}$ random matrices may be chosen as follows:
\begin{itemize}
\item Random matrices $X$ such that $X^p=0$ and so $X$ has
  large kernel. 
\item Random matrices $A$ such that $A^p=\al I$ for a scalar $\al$ and
  hence the inverse of $A$ is easily obtained.
\item It is possible to randomly choose $X,Y$ so that $X^p=0,Y^p=0$
  and so $X,Y$ have large kernel and $(X+Y+1)^p$ or $(X+Y+H)^p$ with
  $h\in C_p^n$ or $(X+Y+W)^p$ for various $w\in C_p^n$ to have value
  $\al I$ for a scalar $\al$.
 \end{itemize} 

Given data $\un{x}=(\al_0,\al_1,\ldots,\al_{m-1})$ of length $m$ we
need this to be
considered in a cyclic group 
ring so  that its completion  is singular of large kernel. 

Let $\al_i\in \Z_2$. 
Consider the group ring $\Z_2C_{2m}$ where $g$ generates $C_{2m}$ and 
let $\un{\un{x}} = \sum_{i=0}^{m-1}\al_ig^i +
\sum_{i=0}^{2m-1}\al_ig^{i+m}$. (Yes, $g^i$ and $g^{i+m}$ have the
same coefficient.) Then $\un{\un{x}}^2 = 0$ and clearly
$\un{x}$ is embedded in $\un{\un{x}}$ and the completion of
$\un{\un{x}}$ has large kernel.

Let $\al_i\in \Z_p$. Consider the group ring $\Z_pC_{pm}$ and suppose
$C_{pm}$ is generated by $g$. Consider
$\un{\un{x}}=\sum_{i=0}^{m-1}\al_ig^i - \sum_{i=0}^{m-1}\al_ig^{i+p}$.
Then $\un{\un{x}}^p = 0$ and $\un{x}$ is embedded in $\un{\un{x}}$.

%The Fourier matrix is very closely related to
 % the representation theory of the cyclic group. Its rows are obtained
 % from the {\em idempotents}. 

%Let $RG$ be the group ring of $G$ over $R$ where $|G|=n$ and $\un{x}
%\in R^n$. Fix a listing of the elements of $G$. 
%The {\em completion of $\un{x}$ in $RG$} is the $RG$-matrix
%obtained with $\un{x}$ as the first row, using this fixed listing.   

\subsubsection{Achieving properties for matrices of general group rings}
For properties of group rings and related algebra consult \cite{sehgal}. The
augmentation mapping $\ep: RG \rightarrow R$ is the ring homomorphism 
 given by 
$\ep(\di\sum_{g\in G}\al_gg) =\di\sum_{g\in G}\al_g$.  Let $R$ be a field
%Now $w$ is non-singular if and only if $W$ is non-singular. 
Suppose now $w$ is non-singular. Then $\ep(w)$ is a unit of $F$ and so
is non-zero. Then $w'=w-\ep(w)1_g$ or $w'=w-\ep(w)g$ for any
$g\in G$ satisfies $\ep(w')=0$ and so $w'$ is singular.  Then $W'$
is singular.  
%Other elements such as $\al g -\al h$ for $g,h\in G$ can
%be added  to $w$ and still retain the singularity. 

Let $\un{x_1}$ be $1\ti r$ data considered as an element of a group ring $FH$
where $F$ is a field. If a key has already been exchanged there is no
need to make the pieces of data singular.

Let $G$ be an $r\ti n$ generator matrix of a zero-divisor $(n,r)$
code over $FH$. Then by Proposition \ref{prop7} the completion $X$ of
$\un{x}=\un{x_1}G$ has rank at most $r$. Thus $\dim \ker X\geq
(n-r)$. Provided $r$ is not very large then given large $n$ it is
impossible to deduce $X$ from $AX$ or $XA$ for an unknown (reasonable)
matrix $A$. For example the code could have large rate say
$\frac{3}{4}$ and then $\dim \ker X \geq \frac{n}{4}$; for $n$
large then also $\dim \ker X$ is large.   This is one way to ensure
the data to be transmitted has large kernel and at the same time
enabling error-correcting.

Thus if $\un{x}$ is data to be transmitted considered as an element of
the group ring $RG$ then $\un{x}-\ep(\un{x})$ is always a singular
element. However this element may have large rank. If this way of
ensuring the data to be transmitted is singular is used then multiple
vector design should be used. The data is broken as
$(\un{x_1},\un{x_2}, \ldots, \un{x_r})$. Then its augmentation is added to
each $\un{x_i}$ to get a vector $\un{y_i}=(\un{x_i},\ep(\un{x_i}))$ 
which is then used. 
%The methods
%requiring data to be singular is then fulfilled and 
So for example 
$(\un{y_1}A_1, \un{y_2}A_2, \ldots, \un{y_r}A_r)$ would be
transmitted. Each piece is singular and $r$ is large.
%the schemes maynot be solved.  

\input{publickey1}

\subsection{Convolution}

Let $\un{z}*\un{w}$ denote the (circulant) convolution of $\un{z}$ and $\un{w}$.
Let  $A$ be a circulant matrix with first row $\un{a}$. 
\begin{lemma}  $\un{x}A = \un{x}*\un{a}$.
\end{lemma} 
\begin{proof} Let $X$ be the completion of $\un{x}$. Then $XA$ is a
  circulant matrix whose first row is $\un{x}A$ and is also
  $\un{x}*\un{a}$.
\end{proof}

More general $G$ convolutions may be defined as follows. 
Let $\un{x} \in R^n, \un{y} \in R^n$ and $G$ a finite group of order $n$.
 Define the $G$-convolution of $\un{x}$ and $\un{y}$, denoted
 $\un{x}*_G\un{y}$, as follows. 
Suppose $\un{x}=(\al_0,\al_1,\ldots,\al_{n-1}),
\un{y}=(\be_0,\be_1,\ldots, \be_{n-1})$ with $\al_i,\be_i \in R$ and
$G=\{g_0,g_1,\ldots,g_{n-1}\}$. Let $x=\al_0g_0+\al_1g_1+\ldots +
\al_{n-1}g_{n-1},y=\be_0g_0+\be_1g_1+\ldots + \be_{n-1}g_{n-1}$. Then
$x\in RG, y\in RG$ and $xy=\ga_0g_0+\ga_1g_1+\ldots +
\ga_{n-1}g_{n-1}$ for some $\ga_i\in R$.
Define $\un{x}*_G\un{y}= (\ga_0,\ga_1,\ldots, \ga_{n-1})$.

\begin{lemma} Let $A$ be an $RG$-matrix with first row $\un{a}$ 
and $\un{x} \in R^n$. Then
  $\un{x}A =\un{x}*_G\un{a}$.
\end{lemma}
\begin{proof} Let $X$ be the completion of $\un{x}$ in $RG$. Then $XA$
  is an $RG$-matrix whose first row is both $\un{x}A$ and
  $\un{x}*_G\un{a}$.
\end{proof}

When $G$ is the cyclic group generated by $g$, 
with listing $\{1,g,g^2,\ldots, g^{n-1}\}$, then $\un{z}*_G\un{w}$ is the
normal (circulant) convolution. 
%Circulant matrices may be diagonalised by the Fourier matrix which
%allow calculations with circulant matrices to be done in $O(n\log n)$ time. 
Calculations in the cyclic group ring and with circulant matrices may be
performed in $O(n\log n)$ time using a fast Fourier transform (FFT)
and FTs allow an effective parallel implementation. 

 %\begin{corollary} Let $A$ be a circulant matrix. Then $\un{x}A$ may be
 %calculated by a fast Fourier transform (FFT) and hence may be
 %performed in $O(n\log n)$ time. 
%\end{corollary}

%When $\un{q}$ denotes the first row of a matrix $Q$ then  
%$\un{q^{-1}}$ denotes the first row of $Q^{-1}$ when it exists.

The encryption methods of the previous sections 
which involve
multiplying vectors and matrices 
can be done in
$O(n\log n)$ time when the matrices have a structure such as the
structure of certain group ring matrices.

\subsection{Coding aspects theory}\label{codingas} Suppose data $x_1$ of size $1\ti r$ is to
be transmitted. Encode $x_1$ by $x_1G=\un{x}$ where $G$ is $r\ti n$
generator matrix of an $(n,r)$ code with $G$ of $\rank r$. 
If $G$ is an $n\ti n$ circulant matrix of $\rank r$ then the first $r$
rows of $G$ are linearly independent; this follows from for the following: 

\begin{lemma} Let $G_1$ be a circulant $n\ti n$ matrix of $\rank r$ and
  suppose $G$ consists of the first $r$ rows of $G_1$. Let $x=x_1G$
  where $x_1$ is a vector of size $1\ti r$ and let $X$ be the
  completion of $x$. Then $\rank X \leq r$.
\end{lemma}
\begin{proof} Let $x_1=(\al_1,\al_2,\ldots,\al_r)$. Then
  $x=x_1G=(\al_1,\al_2,\ldots,\al_r,0,0,\ldots, 0)G_1$ where
  there are $(n-r)$ zeros. Then
  $X=\Gamma G_1$ where $\Gamma$ is the completion of
  $(\al_1,\al_2,\ldots, \al_r,0,0,\ldots,0)$. Hence $\rank X \leq
  \rank G_1 =r$ as required.
\end{proof}
\begin{lemma} Let $G$ be the generator $r\ti n$ 
matrix of a cyclic zero-divisor $(n,r)$ code and $x=x_1G$ where $x_1$
has size $1\ti r$. Then the completion $X$ of $x$ has $\rank \leq r$.
\end{lemma}
\begin{proof} Let the rows of   $G$ be denoted by
$\{\hat{v}_1, \hat{v}_2, \ldots, \hat{v}_r\}$. Then
$x_1G=\di\sum_{i=1}^r\al_i\hat{v}_i$. 
Let $\hat{G}$ be the circulant matrix from which $G$ is derived and
let the rows of this be denoted by $v_1,v_2,\ldots, v_n$. The first
$r$ rows of $\hat{G}$ are linearly independent, see \cite{hur1}, and thus
$\hat{v}_i=\di\sum_{i=1}^r\be_iv_i$ for some $\be_i$. Hence
$x=x_1G=\sum\ga_iv_i$ for some $\ga_i$. 
Thus $x=x_1G=(\ga_1,\ga_2,\ldots, \ga_r, 0,0,\ldots, 0)\hat{G}$ where
$\un{\ga}=(\ga_1,\ga_2,\ldots,\ga_r,0,0,\ldots, 0)$ has length $n$.
Hence $X=\Gamma \hat{G}$ where $X$ is the completion of $x$ and
$\Gamma$ is the completion of $\un{\ga}$. As $\rank \hat{G}=r$ this
implies that $\rank X \leq r$. 
\end{proof}  

More generally we obtain the following result.
\begin{proposition}\label{prop7} Let $G$ be a rank $r$ generator $r\ti n$ matrix of a
zero-divisor $(n,r)$ code obtained from a group ring $n\ti n$ matrix
$\hat{G}$ of $\rank r$.
Then the completion of $x_1G$ in this group ring has $\rank \leq r$.
\end{proposition}
\begin{proof}
Let the rows of $G$ be $\{v_1,v_2,\ldots,v_r\}$ and the rows of $\hat{G}$
be $\{w_1,w_2,\ldots, w_n\}$. Now $\hat{G}$ has $\rank r$ and let $\{w_j
|y\in J\}$ for $J \subset \{1,2,\ldots, n\}$ be a set of $r$ linearly
independent rows of $\hat{G}$.

Now $x_1G=\di\sum_{i=1}^r\al_iv_i$ and $v_i =\di\sum_{j\in J}\be_{i,j}w_j$.
Hence $x_1G=\di\sum_{j\in J}\de_jw_j$. Define for $i=1,2,\ldots,n$,
$\ga_i= \de_i$ when $i\in J$ and $\ga_i=0$ when $i\not \in J$.
Then $x_1G = (\ga_1,\ga_2,\ldots,\ga_n)\begin{ssmatrix}
w_1\\ w_2 \\ \vdots \\ w_n\end{ssmatrix}=(\ga_1,\ga_2,\ldots,\ga_n)\hat{G}$.
Hence $X=\Gamma \hat{G}$ where $X$ is the completion of $x_1G$ and
$\Gamma$ is the completion of $(\ga_1,\ga_2,\ldots,\ga_n)$. As
$\rank \hat{G}=r$ this implies $\rank X \leq r$.
\end{proof}

\end{document}

%% file: publickey.tex
\section{Public key}\label{public}
Public key cryptographic methods may be  designed by choosing vectors with
large kernels from a large pool such that  
a linear combination of these is non-singular. 
The participant A constructs a public key as follows. 
\begin{enumerate}
\item A chooses vectors $\{\un{x},\un{y}\}$ such that their
  completions $\{X,Y\} $
 have large kernels and such that a linear  combination of
 $\{X,Y\}$  is non-singular. 
\item  A chooses non-singular matrices $\{A_1,A_2\}$ and works out
  $\{XA_1,YA_2\}$.
\item A has public key $(XA_1, YA_2)$ and private key $(X,Y,A_1,A_2)$.
\end{enumerate}

Suppose now B wishes to communicate $\un{z}$ to A.
\begin{enumerate}
\item B transmits $(\un{z}XA_1, \un{z}YA_2)$.
\item A works out $(\un{z}X,\un{z}Y)$ and uses the combination $f(X,Y)$ of $X,Y$
  to work out $\un{z}f(X,Y)$ where $f(X,Y)$ is non-singular; from this 
  $\un{z}$ may be worked out by A. 
\end{enumerate}

Methods to  randomly choose such $\{\un{x},\un{y}\}$ 
are shown in section \ref{idemspots}
and methods to randomly choose such $\{A_1,A_2\}$ appear in various
parts of section \ref{theory}. 

The $\un{x},\un{y}$ may be `protected'
on both sides as follows: 
Step 2.\ is replaced by: A chooses $\{A_1,A_2,A_3,A_4\}$ and works out
$(A_1XA_2,A_3YA_4)$; A then has public key $(A_1XA_2,A_3YA_4)$ and
private key \\ $(X,Y,A_1,A_2,A_3,A_4)$. Choosing $A_1$ from a
set of commuting $RG$-matrices and completing the data $\un{z}$ to $Z$
relative to $RG$ enables $ZX$ to be recovered by A from $ZA_1XA_2$ and
similarly $ZY$ may be recovered by choosing $A_3$ from a set of
commuting $RG_1$-matrices where it's not necessary that $G=G_1$.

Details on orthogonal sets of idempotents are given in section \ref{idemspots}. 
Here we outline a method of public key construction 
using full complete orthogonal sets of idempotents.
Let $\{E_0,E_1,\ldots,E_{n-1}\}$ be an complete orthogonal set of idempotents in
$F_{n\ti n}$. Thus here each $E_i$ has $\rank 1$ (but this is not
necessary in general, see section \ref{idemspots}). 
\begin{enumerate} 
\item A chooses $J\subset I$ with $|J|$ approximately half of $|I|=n$ 
  and constructs
  $X=\sum_{j\in J}\al_jE_j, Y=\sum_{j\in (I-J)}\be_jE_j$ with
  $\al_j\neq 0,\be_j\neq 0$. (Here $\rank X = |J|, \rank Y = |I-J|$. 
It is enough to choose $J$ so that both
  $X,Y$ have large kernel.)
%; when $n$ is large $|J|$ could be a
 % fraction of $n$ ensuring such.)
\item A chooses $\{A_1,A_2\}$ non-singular and calculates $\{XA_1,YA_2\}$.
\item A has public key $(XA_1,YA_2)$ and private key $(X,Y,A_1,A_2)$.
\end{enumerate}

When B wishes to communicate $\un{z}$ to A, the process is as
follows.
\begin{enumerate}
\item B transmits $(\un{z}XA_1,\un{z}YA_2)$.
\item A works out $(\un{z}X,\un{z}Y)$ and then $\un{z}(X+Y)$. 
Now $X+Y$ is invertible and its inverse is easy to calculate, by
  Lemma \ref{invert}, and A works out $\un{z}$.
\end{enumerate} 
For each $n$ there are many different complete orthogonal sets of
idempotents in $F_{n\ti n}$. It is not necessary that the particular
complete orthogonal set 
used by A in constructing her public key be known to the world so
in fact an additional step before step 1.\ could be: 
\begin{itemize}
\item[0.] A chooses a complete orthogonal set of idempotents $\{E_0,
  E_1, \ldots, E_{n-1}\}$. 
\end{itemize}

Convolutional methods where appropriate may be used for matrix and
vector-by-matrix multiplications.
The public keys may be changed from time to time. Errors
$(zXA_1+\al,\un{z}YA_2+\be)$ with $\al\neq 0$ or $\be\neq 0$ 
in transmitting $(\un{z}XA_1,\un{z}YA_2)$
are easily detected unless $\al=\ga XA_1$ and $\be=\ga YA_2$ which is
extremely unlikely. This does not prevent an intruder
from trying to falsify a message but a method to prevent this is given in
section \ref{define} below.  
\subsection{From public to private}\label{define} Suppose now A has public key
$(\un{x}A_1,\un{y}A_2)$. This can be made into a private key for B with 
which messages from B only to A may be received:
\begin{itemize}
\item B chooses $\{B_1,B_2\}$ non-singular and transmits $(B_1XA_1,B_2YA_2)$.
\item A chooses $\{A_{B_1},A_{B_2}\}$ and transmits $(B_1XA_{B_1}, B_2YA_{B_2})$.
\item B has key $(XA_{B_1},YA_{B_2})$ with which to send messages to A. 
\end{itemize}
Some simplification is possible when matrices commute.
\begin{itemize}
\item B chooses $\{B_1,B_2\}$ non-singular and transmits
  $(\un{x}A_1B_1,\un{y}A_2B_2)$.
\item A chooses $\{A_{B_1},A_{B_2}\}$ and transmits $(\un{x}A_{B_1}, \un{y}A_{B_2})$.
\item B has (private) 
key $(\un{x}A_{B_1},\un{y}A_{B_2})$ with which to send messages to A.
\end{itemize}
Suppose now B has key $(\un{x}A_{B_1},\un{y}A_{B_2})$ with which to
send message to A. Using this key, B sends message $\un{z}$ to A. Then
A can work out $\un{z}XA_{B_1}$ and check that message has not been
interfered with; an intruder would need to know $XA_{B_1}$ in order to
change message that would not be discovered in a check.
\subsection{Partial public key}\label{define1}
It is useful at times, in particular for authentication and signature
schemes, for a participant to make public a `key' of the form $\un{y}B$
where $\un{y}$ has large kernel and $B$ is invertible. Now $\un{y}B$  
cannot be inverted and so may not be used as a key itself.
It could be used for a message authentication scheme or signature scheme. 

This can be
made private to another particular user by methods similar to those
used in section\ref{define}.
Suppose B has published $\un{y}B$ where $\un{y}$ has large
kernel and $B$ is invertible where  $\{\un{y}, B\}$ are kept private.
\begin{itemize}
\item A chooses $A$ and transmits $AYB$.
\item B chooses $B_A$ and transmits $AYB_A$.
\item A uses $\un{y}B_A$ with B.
\end{itemize}
A simplification using commuting matrices may be initiated similar to
section \ref{define}. 

\subsection{Multiple design for public key} In the above scheme,  vectors 
$\{\un{x},\un{y}\}$ such that their
  completions $\{X,Y\} $
 have large kernels and such that a linear  combination of
 $\{X,Y\}$  is non-singular are chosen. More generally vectors
 $\{\un{x_1},\un{x_2},\ldots, \un{x_r}\}$ such that their
  completions $\{X_1,X_2,\ldots, X_r\} $
 have large kernels and such that a linear  combination of
 $\{X_1,X_2,\ldots,X_r\}$  is non-singular may be chosen. However this
 increases the amount of data to be transmitted as each $\un{z}X_iA_i$
 needs to be transmitted. However again one of these could be laid aside
 authentication; for example a triple of form 
 $(\un{x}A_1,\un{y}A_2.\un{p}A_3)$ each with large kernel such that a
 linear combination of $\{X,Y,P\}$ is non-singular is used but
 $\un{x}A_1=\un{x}A_B$ is private for B only to be used as a check;
 when the message $\un{z}$ is worked out, $\un{z}XA_B$ is used as a
 message authentication check. An original $\un{x}A_1$ may be altered
 to $\un{x}A_B$ by methods similar to those in section \ref{define}.

%% file: publickey1.tex
\subsection{Complete orthogonal sets of idempotents}\label{idemspots}
%Let $\{X,Y\}$ be a complete orthogonal set of idempotents in $F_{n\ti
%  n}$. Restrict $n=2m$ to be even for the moment. Choose $\{X,Y\}$ so
%that $\rank X =\frac{n}{2},\rank Y=\frac{n}{2}$. Then for $n$ large
%both $X,Y$ have large kernel, $\dim \ker X=\frac{n}{2}, \dim \ker
%Y=\frac{n}{2}$. Then $W=\al X +\be Y$ for $\al\neq 0, \be\neq 0$ is
%non-singular and  $W^{-1}=\al^{-1}X+\be^{-1}Y$.

%Let $S=\{E_0,E_1,\ldots, E_{n-1}\}$ be a complete orthogonal set of
%idempotents each of $\rank 1$ in $F_{n\ti n}$. Set $I=\{0,1,\ldots,
%n-1\}$. Choose randomly
%$\frac{n}{2}$ of the elements of $S$, $E_j | i\in J$ where $J\subset
%I$ and $|J| = \frac{n}{2}$.  
Here we consider properties of complete sets of idempotent matrices
and ranks of the idempotents. These are used to construct $X,Y$ such
that these have large kernels and linear combinations of which are
non-singular. 

Let $R$ be a ring with identity $1_R =1$.  A {\em complete family of
orthogonal idempotents} is a set $\{e_1, e_2, \ldots, e_k\}$ in $R$
such that \\ (i) $e_i \not = 0$ and $e_i^2 = e_i$, $1\leq i\leq k$;\\ (ii) If
$i\not = j$ then $e_ie_j = 0$; \\ (iii) $1 = e_1+e_2 + \ldots + e_k$. 

The idempotent $e_i$ is said to be {\em primitive} if it cannot be
written as $e_i= e_i^{'}+ e_i^{''}$ where $e_i^{'},e_i^{''}$ are idempotents
such that $e_i^{'}\neq 0,e_i^{''} \neq 0$ and $e_i^{'}e_i^{''}=0$. A 
set of idempotents is said to be {\em primitive} if each
idempotent in the set is primitive.

For example such sets always exist  in $FG$, the group ring over a
field $F$, when 
$char F \not | \, |G|$; these idempotent sets are
related to the representation theory of $FG$, see \cite{sehgal}.   
General methods for constructing such sets are derived in \cite{hur3} and the
reader is referred therein for details. 
The constructions in \cite{hur3} were derived 
in connection with  applications
 to multi-dimensional paraunitary matrices which are 
 used in the communications' areas. Specific examples of large sets
 and using modular arithmetic (working over $GF(p)$) and where
 convolution methods may be applied  are given in \cite{hur4} . 

For completeness 
 some of the basics are given below.
%\subsection{Determinants and rank}\label{grmat}
%The methods given involves matrices of idempotents. 

%These are used in \cite{hur3} for the construction of paraunitary 
%multi-dimensional matrices; methods for constructing sets from variuos
%systems are derived in this paper. Properties of these are laos
%derived therin and we repeat some of the most relevant ones here. 
 
%\input{rankidem}
%\subsection{Rank and Determinants}\label{grmat1}

\begin{lemma}\label{trrank} Suppose $\{E_1,E_2, \ldots, E_s\}$ is a
set of orthogonal idempotent matrices. Then $\rank
(E_1+E_2 +\ldots + E_s) = \tr (E_1+E_2+ \ldots + E_s) = \tr E_1+ \tr
 E_2+ \ldots + \tr E_s = \rank E_1+ \rank E_2 + \ldots +\rank
E_s$.
% or more generally $\rank (E_{i_1} + E_{i_2} + \ldots E_{i_s}) =
%\rank E_{i_1} +\rank E_{i_2}+ \ldots + \rank E_{i_s}$ 
\end{lemma}
\begin{proof}
It is known that $\rank A = \tr A$ for an idempotent matrix, 
%see for example \cite{idemrank}, 
and so
$\rank E_i = \tr E_i$ for each $i$. If $\{E,F,G\}$ is a set an orthogonal
 idempotent matrices so is  $\{E+F,G\}$. From this it follows that $\rank
(E_1+E_2 +\ldots + E_s) = \tr (E_1+E_2+ \ldots E_s)= \tr E_1+\tr E_2 +
 \ldots + \tr E_s = \rank E_1+ \rank E_2 + \ldots \rank
E_s$.
\end{proof}
\begin{corollary}\label{trrank1}
$\rank(E_{i_1}+ E_{i_2}+ \ldots + E_{i_k})= 
\rank E_{i_1} +\rank E_{i_2}+ \ldots + \rank E_{i_k}$ for $i_j \in \{
1,2,\ldots, s\}$, and $i_j\neq i_l$ for $j\neq l$.
\end{corollary}

Let $\{e_1, e_2, \ldots, e_k\}$ be a complete orthogonal set of idempotents
in a vector space over $F$. %in the group ring $FG$ where $|G| = n$ and let $E_i$ be the matrix
%corresponding to $e_i$ for $i=1,2, \ldots, k$.

%Let $E_{i,j}$ denote the $j^{th}$ column of $E_i$.

\begin{lemma}\label{gr1} Let $w= \al_1 e_1 + \al_2 e_2 + \ldots +
\al_ke_k$ with $\al_i \in F$. Then $w$
 is invertible if and only if each $\al_i \neq 0$ 
and in this case $w^{-1}
 = {\al_1}^{-1}e_1 + {\al_2}^{-1}e_2+ \ldots + {\al_k}^{-1}e_k$.
\end{lemma} 

\begin{proof} Suppose each $\al_i \neq 0$.
Then $w({\al_0}^{-1}e_0+{\al_1}^{-1}e_1+ \ldots + {\al_k}^{-1}e_k)
 = e_0^2 + e_1^2 + \ldots + e_k^2 = e_0+e_1+\ldots + e_k = 1$.

Suppose on the other hand $w$ is invertible and that some $\al_i=0$. 
Then $we_i =0$ and so $w$ is a (non-zero) zero-divisor and is not invertible.
\end{proof}

%\begin{corollary}{\label{vert}}

Now specialise the $e_i$ to be $n\ti n$ matrices and in this case
use capital letters and let $e_i = E_i$.
\begin{lemma}\label{invert} Let $\{E_1,E_2,\ldots, E_k\}$ be a complete orthogonal set
of idempotents in $F_{n\ti n}$ and define  
$A= a_1 E_1 + a_2 E_2 + \ldots + a_kE_k$. Then $A$
 is invertible if and only if each $a_i \neq 0$ and in this case $A^{-1}
 = {a_1}^{-1}E_1 + {a_2}^{-1}E_2+ \ldots + {a_k}^{-1}E_k$. 
\end{lemma}

The reader may consult \cite{hur3} for a proof of the following.
\begin{proposition}{\label{det}} Suppose $\{E_1, E_2, \ldots, E_k\}$ is a
 complete symmetric orthogonal set of idempotents in $F_{n\ti n}$. 
Let $A= a_1 E_1 + a_2 E_2 + \ldots +
 a_kE_k$ with $a_i\in F$.  Then the determinant of $A$ is 
$|A| = a_1^{\rank E_1}a_2^{\rank E_2}\ldots a_k^{\rank E_k}$.
\end{proposition}

\begin{lemma}\label{idemspot} Let $\{E_0,E_1,E_2,\ldots, E_{n-1}\}$ be a complete
  orthogonal set of idempotents in $F_{n\ti n}$ where each $E_i$ has
  $\rank 1$. Let $I=\{0,1,\ldots, n-1\}$ and $J\subset I$. Define
  $X=\di\sum_{j\in J}\al_jE_j$ with $\al_j\neq 0$. Then $\rank X =
  |J|$.
\end{lemma}
\begin{proof}
Let $W=X + \di\sum_{j\in (I-J)}E_j$. Then by Lemma \ref{invert} $W$ is invertible
and so has $\rank n$. Hence
$n=\rank(W)=\rank(X+\di\sum_{j\in (I-J)}E_j) \leq \rank X + \rank
\sum_{j\in (I-J)}E_j = \rank X + (n-|J|)$, by Corollary \ref{trrank1}. Therefore $\rank X \geq |J|$. 
From the rank inequality $\rank (AB)\geq \rank A+\rank B -n$, 
 get $0\geq \rank X + \rank
\di\sum_{j\in (I-J)}E_j-n= \rank X + n-|J|-n$ and hence $|J| \geq
\rank X$. Thus $\rank X = |J|$ 
\end{proof}

The following Lemma may be proved similarly. 
\begin{lemma}\label{idemspot1} Let $\{E_0,E_1,E_2,\ldots, E_{k}\}$ be a complete
  orthogonal set of idempotents in $F_{n\ti n}$ where $\rank
  E_i=r_i$. Let $I=\{0,1,\ldots, k\}$ and $J\subset I$. Define
  $X=\di\sum_{j\in J}\al_jE_j$ with $\al_j\neq 0$. Then $\rank X =
  \sum_{j\in J}\rank E_j$.
\end{lemma}

This enables the construction of public keys as follows.
 
Let $\{E_0,E_1,\ldots,E_{n-1}\}$ be a complete orthogonal set of idempotents in
$F_{n\ti n}$ and $I=\{0,1,\ldots,n-1\}$. 
\begin{enumerate} 
\item A chooses $J\subset I$ with $|J|$ approximately half of $|I|=n$ and constructs
  $X=\di\sum_{j\in J}\al_jE_j, \, Y=\di\sum_{j\in (I-J)}\be_jE_J$ with
  $\al_j\neq 0,\be_j\neq 0$. It is enough to choose $J$ so that both
  $X,Y$ have large kernel. 
\item A chooses $\{A_1,A_2\}$ non-singular and calculates $\{XA_1,YA_2\}$.
\item A has public key $(XA_1,YA_2)$ and private key $(X,Y,A_1,A_2)$.
\end{enumerate}

More generally, proceed as follows to construct a public key for A.

Let $\{E_0,E_1,\ldots,E_{k}\}$ be a complete orthogonal set of idempotents in
$F_{n\ti n}$ and $I=\{0,1,\ldots,k\}$. 
\begin{enumerate} 
\item A chooses $J\subset I$ and constructs
  $X=\di\sum_{j\in J}\al_jE_j, \, Y=\di\sum_{j\in (I-J)}\be_jE_j$ with
  $\al_j\neq 0,\be_j\neq 0$. Then $\rank X = \sum_{j\in J}\rank E_j$
  and $\rank Y=\sum_{j\in(I-J)}E_J=n-\rank X$ 
and $J$ needs to be  chosen so that both
  $X,Y$ have large kernel. 
\item A chooses $\{A_1,A_2\}$ non-singular and calculates $\{XA_1,YA_2\}$.
\item A has public key $(XA_1,YA_2)$ and private key $(X,Y,A_1,A_2)$.
\end{enumerate}

When B wishes to communicate $\un{z}$ to A, the process is  as
follows.
\begin{enumerate}
\item B transmits $\un{z}XA_1,\un{z}YA_2$.
\item A works out $\un{z}X,\un{z}Y$ and then $\un{z}(X+Y)$. 
\item Now $X+Y$ is invertible by Lemma \ref{invert} and easy to
  calculate and A works out $\un{z}$.
\end{enumerate} 

 For each $n$ there are many different complete orthogonal sets of
idempotents in $F_{n\ti n}$. It is not necessary that the particular
set used by A in constructing her public key be known to the world so
in fact an additional step (before 1.\ ) in constructing public key could be: 
\begin{itemize}
\item[0.] A chooses a complete orthogonal set of idempotents $\{E_0,
  E_1, \ldots, E_{k}\}$ in $F_{n\ti n}$. 
\end{itemize}

Schemes where $X,Y$ obtained  from orthogonal sets of idempotents as
above  are `protected' on both sides, as explained in section \ref{public}, may
also be implemented; details are omitted.